\documentclass[aps,prx,showpacs,amsmath,amssymb,amsfonts,lengthcheck,twocolumn,longbibliography,superscriptaddress]{revtex4-2}

\usepackage{derivative}
\usepackage{graphicx}
\usepackage{dsfont}
\usepackage{subfigure}
\usepackage{caption}
\usepackage{verbatim}
\usepackage{ragged2e}
\usepackage{dcolumn}
\usepackage{bm}
\usepackage{epsf}
\usepackage{color}
\usepackage[toc]{appendix}
\usepackage{stmaryrd}
\usepackage{amsthm}
\usepackage{hhline}
\usepackage{amssymb}
\usepackage{mwe,tikz}\usepackage[percent]{overpic}
\usepackage{tikz}
\usetikzlibrary{calc}
\usetikzlibrary{arrows,shapes,calc,matrix}
\usepackage[normalem]{ulem}
\expandafter\let\csname equation*\endcsname\relax
\expandafter\let\csname endequation*\endcsname\relax
\usepackage{amsmath}
\usepackage{xstring}

\newtheorem{theorem}{Theorem}

\newtheorem{lemma}{Lemma}

\makeatletter
\newcommand\footnoteref[1]{\protected@xdef\@thefnmark{\ref{#1}}\@footnotemark}
\makeatother

\newcommand{\bk}[2]     {\langle #1 | #2 \rangle}

\newcommand{\cS}        {{\mathcal S}}

\newcommand{\cE}        {{\mathcal E}}

\newcommand\cF{{\mathcal F}}
\newcommand\hocom[1]{}

\newcommand{\ba}{\begin{eqnarray}}
	\newcommand{\ea}{\end{eqnarray}}
\newcommand{\bmath}{\begin{mathletters}}
	\newcommand{\emath}{\end{mathletters}}
\newcommand{\ban}{\begin{eqnarray*}}
	\newcommand{\ean}{\end{eqnarray*}}

\DeclareRobustCommand{\rchi}{{\mathpalette\irchi\relax}}
\newcommand{\irchi}[2]{\raisebox{\depth}{$#1\chi$}}

\newcommand{\tr}[1]{\mathrm{tr}\left\{#1\right\}}

\newcommand{\bla}{bla\\bla\\bla\\bla\\bla}
\newcommand{\PRA}{Phys. Rev. A }

\usepackage{fmtcount}

\newcommand{\mc}[1]{\mathcal{#1}}

\newcommand{\draftmode}{1}    
\newcommand{\notetoself}[1]{\ifnum \draftmode=1 {\color[rgb]{0,0,0.8} [#1]} \fi}  
\newcommand{\cuttext}[1]{\ifnum \draftmode=1 {\color[rgb]{0,0.5,0} [#1]} \fi}  
\newcommand{\warntext}[1]{\ifnum \draftmode=1 {\color[rgb]{0.9,0.6,0} #1} \else {#1} \color{black} \fi}
\newcommand{\aref}[1]{{Appendix~\hyperref[#1]{A}}}
\newcommand{\bref}[1]{{Appendix~\hyperref[#1]{B}}}
\newcommand{\dref}[1]{{Appendix~\hyperref[#1]{C}}}
\usepackage[colorlinks=true,citecolor=blue,linkcolor=blue,urlcolor=blue]{hyperref}
\usepackage{cleveref}

\begin{document}

\title{Consensus About Classical Reality in a Quantum Universe}
\author{Akram Touil} 
\affiliation{$\hbox{Theoretical Division, Los Alamos National Laboratory, Los Alamos, New Mexico 87545}$
 }
 
 \author{Bin Yan}
\affiliation{$\hbox{Theoretical Division, Los Alamos National Laboratory, Los Alamos, New Mexico 87545}$
 }

 \author{Wojciech H. Zurek}
\affiliation{$\hbox{Theoretical Division, Los Alamos National Laboratory, Los Alamos, New Mexico 87545}$
 }

\date{\today}

\begin{abstract}
Quantum Darwinism recognizes that decoherence imprints redundant records of preferred quasi-classical pointer states on the environment. These redundant records are then accessed by observers. We show how redundancy enables and even implies consensus between observers who use fragments of that decohering environment to acquire information about systems of interest. We quantify consensus using information-theoretic measures that employ mutual information to assess the correlation between the records available to observers from distinct -- hence, independently accessible -- fragments of the environment. We prove that when these fragments have enough information about a system, observers that access them will attribute the same pointer state to that system. Thus, those who know enough about the system agree about what they know. We then test proposed measures of consensus in a solvable model of decoherence as well as in numerical simulations of a many-body system. These results provide detailed understanding of how our classical everyday world arises from within the fundamentally quantum Universe we inhabit.
\end{abstract}
\maketitle
\section{Introduction} 
Quantum Darwinism~\cite{Zurek2000AP,Ollivier2004PRL,Ollivier2005PRA,Zurek2009NP,zurek2022entropy,Korbicz2021roadstoobjectivity,Zurek_2025, davide2022} explains the emergence of objective classical reality from within our quantum universe. It builds on decoherence theory~\cite{Zurek2003RMP,schlosshauer2007,SCHLOSSHAUER2019,joos2013decoherence} which focuses on the suppression of the quantumness and on the environment-induced superselection (einselection) of the preferred pointer states~\cite{basis1,basis2}. Quantum Darwinism goes beyond the decoherence paradigm by recognizing that the decohering environment $\mc{E}$ also serves as a communication channel, carrying information about the state of the system of interest $\mc{S}$. 

For instance, light delivers most of our information. We intercept only a tiny fraction of the photons scattered or emitted by the objects of interest with our eyes, yet we obtain sufficient information about their states. 
This ability to recover information about the system from a small fraction -- fragment $\cF$ -- of its environment means that there are multiple, redundant imprints of the state of $\cS$ on that environment.
Hence, only the states that can survive repeated copying can be perceived in this indirect manner~\cite{Z07a,Z13}: When a quantum system $\mc{S}$ interacts with its environment $\mc{E}$, only the einselected pointer states can survive and produce multiple records of their persistent presence~\cite{basis1,basis2,Zurek_2025}. 

However, even such decoherence-resistant pointer states are still vulnerable: 
Direct measurements of non-commuting observables would re-prepare them and invalidate past records, precluding consensus between observers. Despite this fragility of quantum states, we inhabit a predictable classical world. This is possible because we do not rely on direct measurements. Rather, we monitor systems of interest indirectly, inferring their states from the small fragments of their environment. This avoids the risk of re-preparation of the state of the system, enabling consensus responsible for our perception of objective classical reality. This is also how perception of a unique result (``collapse'') can arise. 

Redundancy allows many observers to reach compatible conclusions about a system. It also implies that consecutive records extracted from the same environment will point to persistent presence of the same einselected states of $\cS$.

Studies on the emergence of classicality have focused on 
models of varying sophistication~\cite{Giorgi2015PRA,Balaneskovic2015EPJD,Balaneskovic2016EPJD,Knott2018PRL,Milazzo2019PRA,Campbell2019PRA,Ryan2020,Garcia2020PRR,Lorenzo2020PRR,Qdc1,Qdc2,Qdc3,Qdc4,Qdc5,Qdc6,Qdc7,Qdc8,Qdc9,Qdc10,Qdc11,Ciampini2018PRA,Chen2019SB,Unden2019PRL,Garcia2020NPJQI,touil2022branching,ftQD}. Consensus was explored only occasionally and qualitatively, i.e., somewhat indirectly. There are now also several experiments that confirm the basic tenets of Quantum Darwinism~\cite{Unden2019PRL,Ciampini2018PRA,Chen2019SB,sq2025}. 

To quantify consensus between observers, we propose information-theoretic measures and test them on a solvable model, showing that preferred pointer states can lead to consensus, unlike their superpositions.
We begin with a concise overview of information theory's role in Quantum Darwinism, laying the groundwork to address the central question of defining and quantifying consensus. We then present our main findings as theorems, offering rigorous evidence for the inevitable emergence of consensus in a quantum universe. Finally, we illustrate consensus through an analytically solvable model as well as numerical simulations of a many-body system.

\section{Branching states and mutual Information}

Quantum Darwinism recognizes that the information about pointer states is available from the fragments of the same environment (e.g., photons) that contributed to decoherence and their einselection. This information can be quantified using 
mutual information between the system $\cS$ and a fragment $\cF$~\cite{Zurek_2025,Ollivier2004PRL,Zurek2009NP,touil2022,touil2023black},
\begin{equation}
I(\cS:\cF)= H_{\cS}+H_{\cF}-H_{\mc{SF}},
\end{equation}
where $H_{X}=-\tr{\rho_{X} \log_2(\rho_{X})}$ is the von Neumann entropy of $\rho_{X}$. When the joint state of $\cS\cE$ has the branching form:
\begin{equation}
|\Psi_{\mc{SE}}\rangle = \sum^{D_{\mc{S}}-1}_{n=0} \sqrt{q_n} |s_n\rangle |\mc{F}_n\rangle |\mc{F}'_n\rangle|{\cE_{/\cF \cF'_n}}\rangle,
\label{branches}
\end{equation}
the mutual information $I(\cS:\cF)$ approaches the entropy of the system for sufficiently large fragments, $I(\cS:\cF) \approx H_\cS$~\cite{Zurek2003RMP,Ollivier2004PRL,blume2005simple,blume2006quantum}. Above, $D_{\mc{S}}$ is the number of pointer states of $\mc{S}$~\cite{basis1,basis2}, while $\mc{F}$ and $\mc{F}'$ represent different fragments and $\cE_{/\cF \cF'}$ is the rest of the environment. The above branching structure of the state may appear to be a strong assumption, but it is justified by appealing to decoherence of $\cS$ by an environment that consists of multiple subsystems (e.g., photons)~\cite{Riedel2010PRL,Riedel2011NJP}.

The mutual information $ I(\cS:\cF)$ is a measure of how much $\cF$ knows about $\cS$. Therefore, one might expect that, when there are many non-overlapping fragments $\cF$ that satisfy $I(\cS:\cF) \approx H_\cS$, observers who gain information about the system from these fragments should agree about the state of $\cS$. This is a suggestive argument, but a more direct assessment of consensus involves mutual information $I(\cF:\cF')$;
\begin{equation}
I(\cF:\cF')= H_{\cF}+H_{\cF'}-H_{\mc{FF'}},
\label{mutFF}
\end{equation}
which quantifies how much two fragments of $\cE$ know about one another. Nevertheless, 
one can imagine that $\cF$ and $\cF'$ share information that has nothing to do with the system of interest. Furthermore, while both $I(\cS:\cF)$ and $I(\cF:\cF')$ quantify shared information, they answer the question ``how much'', but---in absence of additional assumptions---they bypass the question ``what observable of $\cS$ is this information about?''. This concern---obvious for $I(\cF:\cF')$---can be also raised about $I(\cS:\cF)$: Granted, information is then about $\mc{S}$, but what states of the system are imprinted on, and presumably accessible via indirect measurements on the fragments of the environment? Last but not least, mutual informations provide answers in bits, about the shared information. However, consensus should be quantified by the extent of the agreement between the data available to different observers, e.g., by a number in the interval $[0,1]$, where $0$ and $1$ would correspond to its absence and to perfect consensus, respectively.

{\it Quantifying consensus via mutual information.}
To quantify the quality of the record of $\cS$ in $\cF$ by a number ${\mathfrak c} \in [0,1]$, for the states in~Eq.~\eqref{branches}, we normalize $I(\cS:\cF)$:
\begin{equation}
{\mathfrak c}(\cF : \cS) =  I(\cS:\cF) / H_\cS.
\label{consen1}
\end{equation}
This is a measure of the consensus between the state of the system $\cS$ and its record in the fragment $\cF$. For branching states, the plot of $I(\cS:\cF)$ asymptotes to the classical plateau, $ I(\cS:\cF) \approx H_\cS $, when the fragment size is large. It follows that ${\mathfrak c}(\cF : \cS)$  will asymptote to $1$, indicating the presence of a complete record in $\cF$ of the classical state of the system \footnote{We note that when $\cS\cE$ is pure, and $\cF $ approaches $\cE$, ${\mathfrak c}(\cF : \cS)_\cS $ can approach 2 as a consequence of entanglement that is responsible for the rise of $ I(\cS:\cF)$ to $2 H_\cS$ as $\cF \rightarrow \cE$.}.  

A measure that addresses a related question employs a hybrid (Shannon - von Neumann) entropy:
\begin{equation} 
{\mathfrak c}( \cS_\Lambda : \cF) = \rchi(\mc{F}:{\mc{S}}_\Lambda) / H_{ \cS} .
\end{equation}
It quantifies correlation between $\cF$ and an observable ${\Lambda}$ of the system as indicated by “$ \cS_\Lambda $”. Thus, ${\mathfrak c}(\cS_\Lambda : \cF)$ involves an observable-dependent conditional entropy: The Holevo quantity $\rchi(\mc{F}:{\mc{S}}_\Lambda)$ is defined by;
\begin{equation}
\rchi(\mc{F}:{\mc{S}}_\Lambda)=H_{\cF}-H_{\cF|{{\cS}_\Lambda}}.
\end{equation}
Above, $H_{\cF|{ \cS}_\Lambda}=-\sum_{i} p_{i} \tr{\rho^{(i)}_{\cF} \log_2(\rho^{(i)}_{\cF})}$ is the average entropy of the fragment $\cF$ after measurements of $\Lambda$ on $\cS$, where each measurement outcome leads to the conditional state $\rho^{(i)}_{\cF}$ with probability $p_i$.


The consensus between two fragments is defined by;
\begin{equation}
{\mathfrak c}(\cF :\cF') =  I(\cF:\cF') / H_\cS.
\label{consen2}
\end{equation}
This is the fraction of information about which $\cF$ and $\cF'$ agree: The correlation in the information content of $\cF$ and $\cF'$ is quantified by the mutual information $I(\cF:\cF')$, the key ingredient of ${\mathfrak c}(\cF:\cF')$. 

\section{sufficient information imposes consensus}

The mutual information $I(\cF:\cS)$ quantifies how much the fragment $\cF$ knows about the system. It increases with the size of $\cF$, approaching an approximately constant value $H_\cS$ at the {\it classical plateau} (see, e.g., Fig.~\ref{fig01}). Exact equality would indicate that $\cF$ has all the classically accessible information about the system. This can be attained only for fragments as large as half of the environment, $|\cF|=|\cE|/2$. We are interested in consensus between many observers who access independent fragments, so $|\cF|\ll|\cE|$. Mutual information $I(\cF:\cS)$ in such fragments can approach the classical plateau. The information deficit $\delta$~\cite{Ollivier2004PRL} quantifies how close for a fragment $\cF$, $I(\cF:\cS)$ is to the plateau:
$$I(\cF:\cS) = (1-\delta) H_\cS. $$
Small $\delta$ means $\cF$ is enough to extract most (all except for the {\it information deficit} $\delta$) of the classical information about $\cS$. 

We shall now demonstrate that small information deficit implies consensus about the state of $\cS$.
Consider a completely decohered branching state: When the reminder of the environment eliminates superposition between branches (e.g., when in Eq.~\eqref{branches}
$\bk { \cE_{/\cF \cF'_m}} { \cE_{/\cF \cF'_n}} = \delta_{mn}$),
several entropies that define $I(\cF:\cS)$ and $I(\cF:\cF')$ are equal; 
$H_{\mathcal{S F}}=H_{\mathcal{S F'}}=H_{\mathcal{S F F'}}=H_{\mathcal{S}}$. 
The proof below employs this equality to obtain a bound on {\it consensus deficit} $\hat \delta$ in $I(\cF:\cF') = (1- \hat \delta)H_\cS$ using information deficits of $\cF$ and $\cF'$. 

\begin{theorem}

    
When branches are completely decohered, and $I(\mathcal{F}: \mathcal{S})=(1-\delta) H_{\mathcal{S}}, I(\mathcal{F}^{\prime}:\mathcal{S}) =\left(1-\delta^{\prime}\right) H_{\mathcal{S}}$, there exists $\tilde{\delta}$ such that the mutual information between the fragments $\mathcal{F}$ and $\mathcal{F}^{\prime}$ is bounded from below, $I(\mathcal{F}: \mathcal{F}^{\prime}) =$ $\left(1-\delta-\delta^{\prime}+\tilde{\delta}\right) H_{\mathcal{S}}$, with $\min( \delta, \delta^{\prime}) \geq \tilde{\delta} \geq 0$. The consensus deficit is then $\hat \delta = \delta+\delta^{\prime}-\tilde{\delta}$.
\end{theorem} 

\begin{proof}
The proof is based on the branching structure of the state $\mc{SE}$,~Eq.~\eqref{branches}. Pointer states $|s_n\rangle$ of $\mathcal{S}$ are orthogonal. We also assume that the reminder of the environment $\cE_{/{\cF\cF'}}$ (that is, $\mathcal{E}$ less the composite fragment $\mathcal{F} \mathcal{F}^{\prime}$) is, e.g., large enough to remove off-diagonal terms in the pointer basis representation of $\mathcal{S}$. Then $H_{\mathcal{S F}}=H_{\mathcal{S F ^ { \prime }}}=H_{\mathcal{S F F ^ { \prime }}}=H_{\mathcal{S}}$, and
$$
H_{\mathcal{F}}=(1-\delta) H_{\mathcal{S}}, \quad H_{\mathcal{F}^{\prime}}=\left(1-\delta^{\prime}\right) H_{\mathcal{S}}.
$$
Moreover, as $I(\mathcal{F \mathcal { F } ^ { \prime }}: \mathcal{S})=(1-\tilde{\delta}) H_{\mathcal{S}}$, we get
$$
H_{\mathcal{F} \mathcal{F}^{\prime}}=(1-\tilde{\delta}) H_{\mathcal{S}},
$$
where $\tilde{\delta}$ is the information deficit of a fragment representing the union of $\mathcal{F}$ and $\mathcal{F}^{\prime}$. As $\left|\mathcal{F} \mathcal{F}^{\prime}\right| \geq \max \left(|\mathcal{F}|,\left|\mathcal{F}^{\prime}\right|\right)$ it follows that $\tilde{\delta} \leq \min \left(\delta, \delta^{\prime}\right)$. Hence;
$$
I(\mathcal{F}: \mathcal{F}^{\prime})=\left(1-\delta-\delta^{\prime}+\tilde{\delta}\right) H_{\mathcal{S}},
$$
or
$$
I(\mathcal{F}: \mathcal{F}^{\prime})=(1-\hat \delta) H_{\mathcal{S}},
$$
where $\hat \delta = \delta+\delta^{\prime}-\tilde{\delta}$ can be viewed as consensus deficit between the two fragments.
\end{proof}
We conclude that when two fragments $\cF$ and $\cF'$ have information sufficient to infer the state of the system, they will agree on what that state is.
The above result depends only on the branching structure of the state of $\mathcal{S E}$, and on the assumption that the ``reminder of the environment'' decoheres $\mathcal{S}$, so that $H_{\mathcal{S F}}=H_{\mathcal{S} \mathcal{F}^{\prime}}=H_{\mathcal{S F F}}=H_{\mathcal{S}}$. The consensus ${\mathfrak c}(\cF :\cF')$ between the records in the environment fragments is then directly related to their information deficits, regardless of the size of these deficits. 


Note that discussions of Quantum Darwinism often rely on the assumptions that fragments are ``typical'', and that the information about $\mathcal{S}$ in each fragment depends only on its size - on the number of the environment subsystems. The proof above does not make such assumptions. 


\section{Mutual information and consensus in a simple model}
To illustrate our conclusions we revisit the {\tt c-maybe} model~\cite{touil2022} - a single qubit (system $\mathcal{S}$) coupled to non-interacting qubits in the environment $\mathcal{E}$ via imperfect {\tt c-not} gates, as depicted in Fig.~\ref{fig:enlarged_system}. This model is analytically solvable. Surprisingly, in spite of its simplicity, it results in expressions for mutual information $I(\cS : \cF)$ that appear in realistic photon scattering models~\cite{Riedel2010PRL,Riedel2011NJP,touil2022}. 
\begin{figure}[h!]
\includegraphics[scale=0.38]{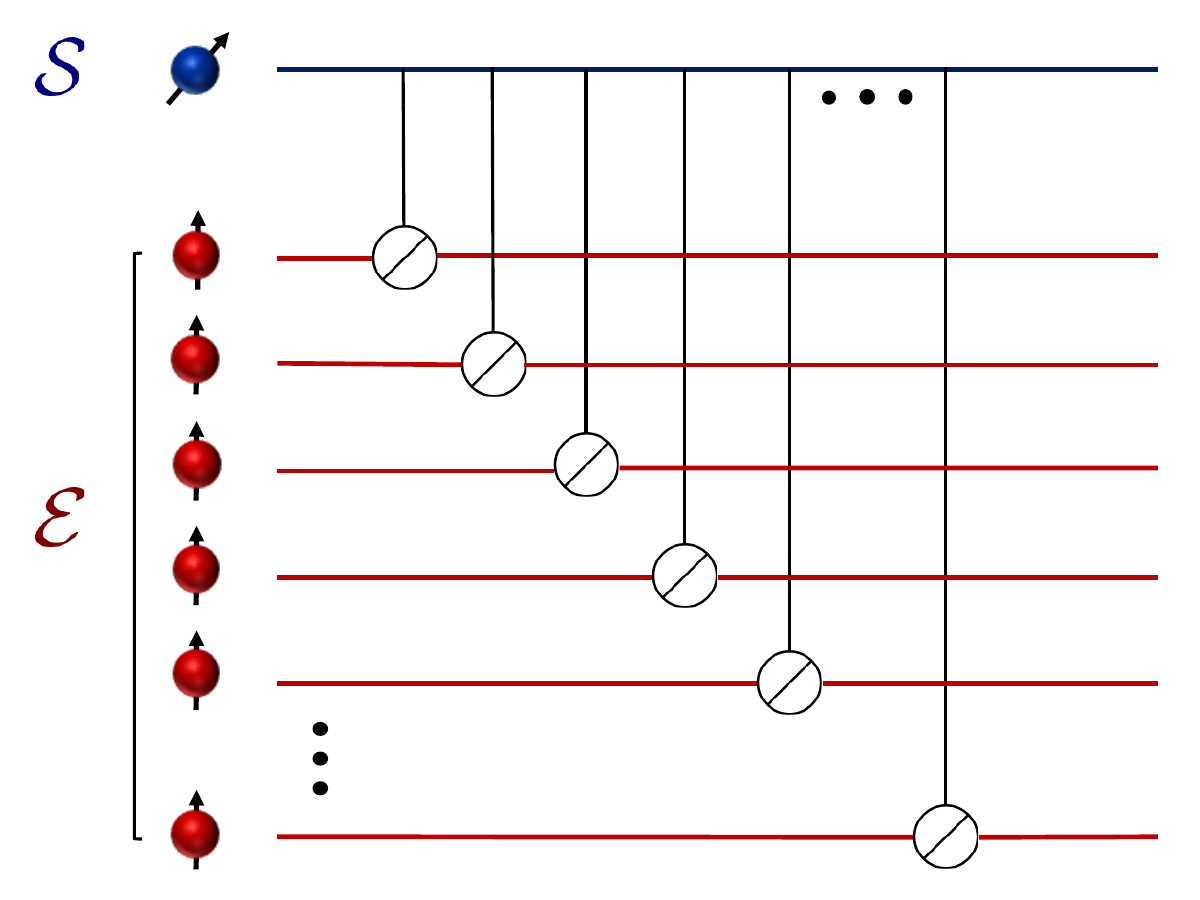}
\caption{\justifying
The interactions between the system and each qubit of the environment in the {\tt c-maybe} model~\cite{touil2022}. The system starts in a superposition of pointer states $|0\rangle$ and $|1\rangle$, and the environment qubits start in the all-up state. In contrast to perfect {\tt c-not} gates, the {\tt c-maybe}  gates $U_{\oslash}$ (illustrated with the ``$\oslash$'' symbol) rotate the states of the environment qubits into a superposition of $|0\rangle$ and $|1\rangle$ (see Eq.~\eqref{ope}). This leads to the branching states that support the emergence of classicality~\cite{touil2022branching}.}
\label{fig:enlarged_system}
\end{figure}
\begin{figure*}[ht]
    \includegraphics[width=\textwidth]{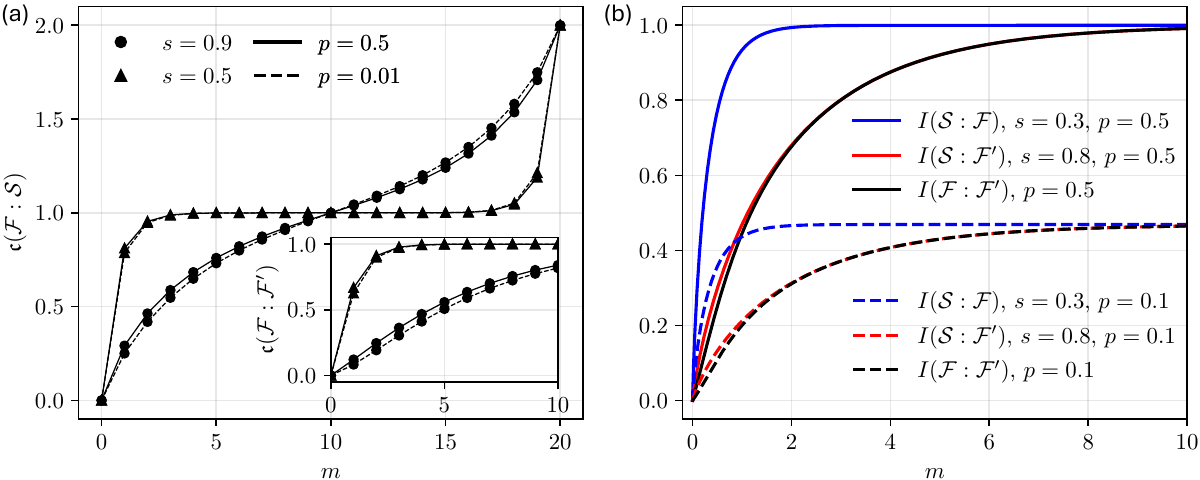}
    \caption{\justifying (a) Correlation measure ${\mathfrak c}(\cF : \cS) = I(\cS:\cF) / H_\cS$ and consensus measure ${\mathfrak c}(\cF : \cF') = I(\cF : \cF') / H_{\cS}$ (bottom right inset) as functions of qubits $m$ in fragment $\mc{F}$ for the {\tt c-maybe} model with various probabilities $p$ (see Eq.~\eqref{cmaybe}) and interaction parameters $s$. The parameter $s$ quantifies information transfer between $\cS$ and $\cE$, with $s=0$ indicating perfect transfer (generalized GHZ state) and $s=1$ indicating no coupling. Here, we consider two disjoint fragments $\mc{F}$ and  $\mc{F}'$ such that $|\mc{F}|=|\mc{F}'|=m$. For this plot, the environment is composed of $N=20$ spins. (b) Plots of mutual informations $I(\mc{S}:\mc{F})$, $I(\mc{S}:\mc{F}')$, and $I(\mc{F}:\mc{F}')$ where $\mc{F}$ and $\mc{F}'$ are non-identical fragments. Information transfer between $\mc{S}$-$\mc{F}$ and $\mc{S}$-$\mc{F}'$ is quantified by $s=0.3$ and $s=0.8$, respectively. These plots show how the rise to the plateau in $I(\mc{S}:\mc{F})$ and $I(\mc{S}:\mc{F}')$ leads to the rise of consensus $I(\mc{F}:\mc{F}')$ without assuming fragment typicality. For this plot, the environment is composed of $N=100$ spins.}
\label{fig01}
\end{figure*}
\begin{figure}[ht]
    \includegraphics[width=\linewidth]{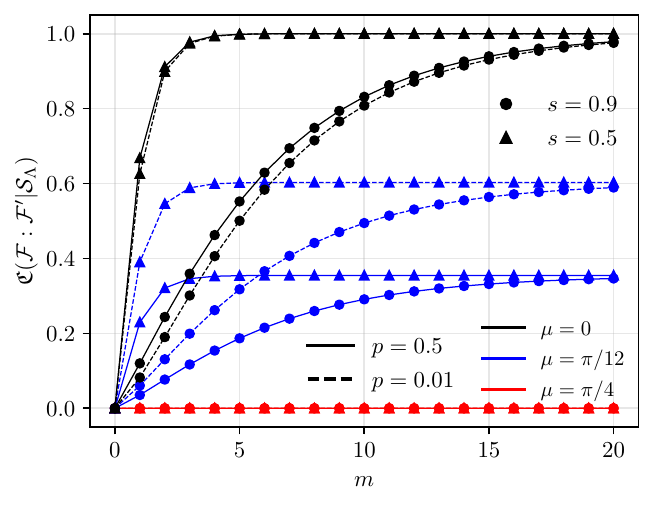}
    \caption{\justifying Plots of the consensus measure ${\mathfrak C}(\cF : \cF' | \cS_\Lambda)$ for different $s$ and $p$, and projections on $\mc{S}$. Specifically, $\mu=0$ corresponds to projections onto the pointer basis, $\mu=\pi/12$ to a basis rotated by the angle $\pi/6$, and $\mu=\pi/4$ to the complementary basis. Projections with $\mu=\pi/4$ reveal zero Holevo information~\cite{Zwolak2013SR}. In particular we project the system onto the basis given by $|+\rangle=$ $\cos (\mu)|0\rangle+ \sin (\mu)|1\rangle$ and $|-\rangle=\sin (\mu)|0\rangle- \cos (\mu)|1\rangle)$. Here, we consider two disjoint fragments $\mc{F}$ and  $\mc{F}'$ such that $|\mc{F}|=|\mc{F}'|=m$ and the environment is composed of $N=100$ spins.}
    \label{fig02}
\end{figure}

The system $\mc{S}$ is a qubit coupled to $N$ independent non-interacting qubits of the environment  $\mc{E}$ via a {\tt c-maybe} gate with the truth table,
\begin{equation}
U_{\oslash}=\begin{pmatrix}
1 & 0 & 0 & 0 \\
0 & 1 & 0 & 0 \\
0 & 0 & s & c \\
0 & 0 & c & -s
\end{pmatrix}.
\label{ope}
\end{equation}
The parameters $c=\cos(a)$ and $s=\sin(a)$ quantify the imperfect transfer of information between system and environment, such that $a$ is the angle associated with the action by which the target qubit is rotated. Now, we assume that the system starts in the initial state $\sqrt{p} |0\rangle + \sqrt{1-p} |1\rangle $ and the environment state is initialized in the all-up state $|00...0\rangle$. After the coupling, the joint $\mathcal{SE}$ state then assumes the branching form:
\begin{equation}
|\Psi_{\mc{SE}}\rangle =  \sqrt{p }|0\rangle \bigotimes^{N}_{k=1} |\varepsilon^{k}_0\rangle + \sqrt{1-p }|1\rangle \bigotimes^{N}_{k=1} |\varepsilon^{k}_1\rangle,
\label{cmaybe}
\end{equation}
where $|\varepsilon_{n}^{k}\rangle$ are individual sub-environment states representing the state of each qubit in the environment, such that $|\langle \varepsilon_{m}^{k}|\varepsilon_{n}^{k}\rangle|=s$ when $m\neq n$. In Ref.~\cite{touil2022}, it was shown that 
\begin{equation}
 I\left(\mathcal{S}: \mathcal{F}\right)=h\left(\lambda_{N, p}^{+}\right)+h\left(\lambda_{m, p}^{+}\right)-h\left(\lambda_{N-m, p}^{+}\right),
\end{equation}
where $m=|\mc{F}|$, $h(x)=-x \log_{2}(x)-(1-x) \log_{2}(1-x)$, $\lambda_{k, p}^{ \pm}=\frac{1}{2}\left(1 \pm \sqrt{(q-p)^2+4 s^{2 k} p q}\right)$ and $q=1-p$. Hence, near the plateau we have 
\begin{equation}
I\left(\mathcal{S}: \mathcal{F}\right) \approx  H_{\mc{S}}- \frac{pq}{|q-p|}|\log_{2}(\frac{q}{p})|s^{2m}.
\end{equation}
and;
\begin{equation}
I\left(\mathcal{S}: \mathcal{F}'\right) \approx H_{\mc{S}}- \frac{pq}{|q-p|}|\log_{2}(\frac{q}{p})|s^{2m'},
\end{equation}
where $m'=|\mc{F}'|$. Therefore,
\begin{equation}
\begin{split}
I\left(\mathcal{F}: \mathcal{F}^{'}\right) &\approx  H_{\mc{S}}\\ &- \frac{pq}{|q-p|}|\log_{2}(\frac{q}{p})|\left(s^{2m}+ s^{2m^{'}}-s^{2(m+m^{'})} \right).
\end{split}
\end{equation}
The above expression explicitly shows the dependence of the information deficits $\delta$, $\delta'$, and $\tilde{\delta}$ on the parameters of the {\tt c-maybe} model, i.e., $s$ and $m$. In particular, we get
\begin{equation}
\begin{split}
\delta&= \left(\frac{pq}{H_{\cS} |q-p|}|\log_{2}(\frac{q}{p})| \right) s^{2m}, \\
\delta'&= \left(\frac{pq}{H_{\cS} |q-p|}|\log_{2}(\frac{q}{p})| \right) s^{2m'}, \\
\tilde{\delta}&= \left(\frac{pq}{H_{\cS} |q-p|}|\log_{2}(\frac{q}{p})| \right) s^{2(m'+m)}.\\
\end{split}
\end{equation}
The resulting behavior of the consensus measures ${\mathfrak c}(\cF :\cF')$ and ${\mathfrak c}(\cS :\cF)$ are illustrated in Fig.~\ref{fig01}.


To establish Theorem 1 we have assumed that state $\cS\cE$ has branching structure, and that the remainder of the environment $\cE_{/{\cF\cF'}}$ (that is, $\mathcal{E}$ less the composite fragment $\mathcal{F} \mathcal{F}^{\prime}$) is sufficiently large to decohere these branches, so that $H_{\mathcal{S F}}=H_{\mathcal{S F ^ { \prime }}}=H_{\mathcal{S F F ^ { \prime }}}=H_{\mathcal{S}}$. For the {\tt c-maybe} model we have the following analytic expressions:
\begin{equation}
\begin{split}
H_{\mc{SF}}&=h\left(\lambda_{N-m, p}^{+}\right),\\
H_{\mc{SF}'}&=h\left(\lambda_{N-m', p}^{+}\right),\\
H_{\mc{SFF}'}&=h\left(\lambda_{N-(m+m'), p}^{+}\right),
\end{split}
\end{equation}
while
\begin{equation}
H_{\mc{S}}=h\left(\lambda_{N, p}^{+}\right).
\end{equation}
Therefore, when the environment is not sufficiently large for its remainder to decohere $\cS\cF\cF'$, this will not be true. For a large, but not infinite environment, there exists $\epsilon$, $\epsilon'$ and $\tilde{\epsilon}$ such that
\begin{equation}
\begin{split}
H_{\mc{SF}}&=(1-\epsilon)H_{\mc{S}},\\
H_{\mc{SF}'}&=(1-\epsilon')H_{\mc{S}},\\
H_{\mc{SFF}'}&=(1-\tilde{\epsilon})H_{\mc{S}}.
\end{split}
\label{epsi}
\end{equation}
In fact, we have
\begin{equation}
\begin{split}
\epsilon &=\frac{pq}{H_{\cS}|q-p|}|\log_{2}(\frac{q}{p})|\left(s^{2N}-s^{2(N-m)}\right),\\
\epsilon' &=\frac{pq}{H_{\cS}|q-p|}|\log_{2}(\frac{q}{p})|\left(s^{2N}-s^{2(N-m')}\right),\\
\tilde{\epsilon} &=\frac{pq}{H_{\cS}|q-p|}|\log_{2}(\frac{q}{p})|\left(s^{2N}-s^{2(N-m-m')}\right).\\
\end{split}
\end{equation}
This leads to the following general Lemma 1:
\begin{lemma}
    For a finite-size environment, when $I(\mathcal{F}: \mathcal{S})=(1-\delta) H_{\mathcal{S}}, I(\mathcal{F}^{\prime}:\mathcal{S}) =\left(1-\delta^{\prime}\right) H_{\mathcal{S}}$, there exists $\hat{\delta}$ and $\hat{\epsilon}$ such that the mutual information between the fragments $\mathcal{F}$ and $\mathcal{F}^{\prime}$ is bounded from below, $I(\mathcal{F}: \mathcal{F}^{\prime}) =$ $\left(1-\hat{\delta}-\hat{\epsilon}\right) H_{\mathcal{S}}$.
\end{lemma} 

\begin{proof}
The proof is based on the branching structure of the state $\mc{SE}$,~Eq.~\eqref{branches}. Unlike the proof for Theorem 1, we only assume that the pointer states $|s_n\rangle$ of $\mathcal{S}$ are orthogonal. In other words, the reminder of the environment $\cE_{/{\cF\cF'}}$ is not large enough to remove off-diagonal terms in the pointer basis representation of $\mathcal{S}$. Then, from~Eq.~\eqref{epsi}, we get
$$
H_{\mathcal{F}}=(1-\delta-\epsilon) H_{\mathcal{S}}, \quad H_{\mathcal{F}^{\prime}}=\left(1-\delta^{\prime}-\epsilon^{\prime}\right) H_{\mathcal{S}}.
$$
Moreover, as $I(\mathcal{F \mathcal { F } ^ { \prime }}: \mathcal{S})=(1-\tilde{\delta}) H_{\mathcal{S}}$, we get
$$
H_{\mathcal{F} \mathcal{F}^{\prime}}=(1-\tilde{\delta}-\tilde{\epsilon}) H_{\mathcal{S}}.
$$
As $\left|\mathcal{F} \mathcal{F}^{\prime}\right| \geq \max \left(|\mathcal{F}|,\left|\mathcal{F}^{\prime}\right|\right)$ it follows that $\tilde{\delta}+\tilde{\epsilon} \leq \min \left(\delta+\epsilon, \delta^{\prime}+\epsilon^{\prime}\right)$. Hence,
$$
I(\mathcal{F}: \mathcal{F}^{\prime})=\left(1-\delta-\epsilon-\delta^{\prime}-\epsilon^{\prime}+\tilde{\delta}+\tilde{\epsilon}\right) H_{\mathcal{S}}.
$$
or
$$
I(\mathcal{F}: \mathcal{F}^{\prime})=\left(1-\hat{\delta}-\hat{\epsilon}\right) H_{\mathcal{S}}.
$$
\end{proof}

 Theorem 1 was deduced under the assumption that $H_{\mathcal{S F}}=H_{\mathcal{S F ^ { \prime }}}=H_{\mathcal{S F F ^ { \prime }}}=H_{\mathcal{S}}$.  When the environment is in effect infinite (as in everyday settings, as recognized by the assumptions of Theorem 1) the conditions that lead to the limit on the consensus deficit are met. Lemma1 explores the consequences of relaxing of this assumption. 

\section{The refined Mutual Information and consensus} 

We have seen that, for branching states, the mutual information-based measure introduced in the last section can successfully quantify the consensus between two environment fragments. 
Branching states~(\ref{branches}) are a useful simplifying assumption that allows one to capture the essence of information flows relevant for Quantum Darwinism and for the transition from quantum to classical. However, branching states ignore the possibility that the environment fragments may be correlated with systems other than $\cS$, or simply with each other when there is direct interaction between them~\cite{touil2022,duruisseau2023pointer}. This raises the possibility that the information they share may be about something else than the system of interest $\cS$. 


To determine whether the information in $\cF$ and $\cF'$ leads to the same conclusion about the state of $\cS$ we consider {\it refined mutual information} defined as the difference between the mutual information and the conditional mutual information:

\begin{equation}
\mathfrak{I} (\cF : \cF' | \cS) = I (\cF : \cF') - I(\cF : \cF' | \cS) .
\label{refine}
\end{equation}
Refined mutual information is sometimes called ``interaction'' in information theory~\cite{CoverThomas}. For obvious reasons we will not adopt this nomenclature in a physics paper. Our motivation for considering ${\mathfrak I} (\cF : \cF' | \cS)$ is clear: The information about $\cS$ in the fragments of its environment is the focus of Quantum Darwinism. Yet, $ I (\cF : \cF') $ is {\it all} the information shared by the two fragments, so it may include spurious information, that is ``not just about $\cS$''. For the states in~Eq.~\eqref{branches}, the conditional mutual information:
\begin{equation}
I(\cF : \cF' |  \cS) = H_{\cF| \cS} + H_{\cF' |  \cS} - H_{\cF\cF'| \cS} 
\label{condmut}
\end{equation}
is the information shared by $\cF$ and $\cF'$ that is {\it not} about a certain observable of $\cS$: When $\cF$ and $\cF'$ are part of a branching state~\eqref{branches}, they share information about the pointer observable of $\cS$ and nothing else. Therefore, knowing the state of $\cS$ reveals their states: Their conditional entropies vanish,
\begin{equation}
H_{\cF|\cS} = H_{\cF'|\cS} = H_{\cF\cF'| \cS} = 0,
\end{equation}
since their conditional states are pure. Thus, in a branching state, $\cF$ and $\cF'$ know only about the pointer observable of $\cS$. Hence $I(\cF : \cF' | \cS)=0$ and 
$ {\mathfrak I} (\cF : \cF' | \cS) = I (\cF : \cF').$
By contrast, when $\cF$ and $\cF'$ are correlated, but their states are not correlated with $\cS$, conditioning on the state of $\cS$ has no effect, so that $ I (\cF : \cF') = I(\cF : \cF' | \cS)$. Intermediate cases (e.g., where some of the information is about $\cS$) are also possible~\cite{darwin5}. 

Therefore, a measure of consensus about the state of the system that is not limited in its applicability to branching states is based on the refined mutual information:
\begin{equation}
{\mathfrak C }(\cF : \cF' | \cS) =
{\mathfrak I} (\cF : \cF' | \cS) 
/ H_\cS.
\end{equation}
For branching states $I(\cF : \cF' | \cS)=0$, and ${\mathfrak c}(\cF : \cF') = {\mathfrak C }(\cF : \cF'| \cS)=I(\cF:\cF') / H_\cS$ for projections of $\mc{S}$ onto the pointer basis.

Above we have assumed an optimal measurement 
which eliminates the need to specify the observable measured on $\cS$ explicitly. Observers have access to arbitrary measurements which may involve POVMs. One could also consider:
\begin{equation}
{\mathfrak C }(\cF : \cF' |  \cS_\Lambda)= \
{\mathfrak I} (\cF : \cF' | \cS_\Lambda)
/ H_\cS.
\end{equation}
The observable-dependent refined mutual information quantifies consensus between $\cF$ and $\cF'$ about an observable ${\Lambda}$ of the system $\cS$, as indicated by “$\cS_\Lambda $”. It involves an observable-dependent conditional entropy. 

The conditional mutual information required for the definition of ${\mathfrak C }(\cF : \cF'|{\cS}_\Lambda)$ presupposes a measurement. Hence, it will depend on the measured observable. When $\cF$ and $\cF'$ in the perfect branching state,~Eq.~\eqref{branches}, are conditioned on the pointer observable $\Pi$ (with eigenstates $|s_n\rangle$) of the system, their states are pure. Therefore, $ H_{\cF|{\cS}_\Pi} = H_{\cF'|{\cS}_\Pi} = H_{\cF\cF'|{\cS}_\Pi} = 0$, so that $I(\cF : \cF' | {\cS}_\Pi)=0$. The consensus ${\mathfrak C }(\cF : \cF'|{\cS}_\Pi)$ is then maximized. By contrast, when an observable complementary to $\Pi$ is measured on $\cS$, it does not reveal anything about the states of the fragments \footnote{it might reveal the state of the whole environment, however!}. In that case $I(\cF : \cF' | {\cS})= I(\cF : \cF')$, so there is no consensus about an observable complementary to the pointer observable $\Pi$. Conditioning on other observables of $\cS$ leads to intermediate values of consensus. 

Refined mutual information ${\mathfrak I }(\cF : \cF' | \cS_\Lambda)$ can take negative values. This happens when conditioning increases the mutual information between the two fragments. This might complicate 
its interpretation as a measure of consensus between fragments. However, 
when wavefunction admits the branching structure, refined mutual information is always positive under the assumption that the rest of the environment and one of the fragments ($\cF$ or $\cF'$) have orthogonal states in the different branches of the wave function (i.e. large enough to ensure perfect record states). This is established by the theorem.

\begin{theorem}
    Consider the singly-branching structure,~Eq.~\eqref{branches} of the global wave function $|\psi_{\mc{SE}}\rangle$. Independent of the measurements applied on $\mc{S}$, refined mutual information ${\mathfrak I }(\cF : \cF' | \cS_\Lambda)$ is always positive when the fragment $\mc{F}$ and the rest of the environment $\cE_{/\cF \cF'}$ have orthogonal states in the different branches of the wave function $|\psi_{\mc{SE}}\rangle$.
\end{theorem}
Proof of this theorem is delegated to the appendix. Additionally, to illustrate this theorem we plot consensus ${\mathfrak C }(\cF : \cF' | \cS_\Lambda)$ for the {\tt c-maybe} model in Fig.~\ref{fig02}.

However, in more realistic scenarios, the environment's degrees of freedom might interact with one another, mixing the information they acquire about $\mc{S}$. To explore such cases, we perform numerical simulations using an all-to-all spin model, where our consensus measure ${\mathfrak C }(\cF : \cF' |  \cS_\Lambda)$ can become negative, yet still serves as a useful indicator of agreement between observers regarding the system observable $\Lambda$.

\section{Rise and Fall of Consensus: Decoherence and Relaxation}

Branching states
are a simplifying assumption. It is correct in some important cases (e.g., photons), but it is interesting to investigate what happens when it is only an approximation. Here we relax it by allowing the subsystems of $\cE$ to interact. The correlation of the two fragments $\cF$ and $\cF’$ will then reflect not only what they ``know’’ about $\cS$, but also information they acquire by direct interaction with one another. Our task is to distinguish and quantify the mutual information corresponding to consensus about the state $\cS$---the basis of the objective classical reality of its state---from the mutual information $\cF$ and $\cF’$ have about each other that is unrelated to what they ``know’’ about $\cS$.

We have already introduced refined mutual information $\mathfrak{I}(\cF:\cF’ |\cS)$,~Eq.~\eqref{refine}, the information-theoretic tool that makes this possible. Here we illustrate refined mutual information on a simple example in which subsystems of the environment interact with each other as well as with $\cS$. The complete Hamiltonian is given by:

\begin{equation}
\boldsymbol{H} = \boldsymbol{\sigma}_{\mathcal{S}}^z \otimes \sum_{i=1}^N d_i\,  \boldsymbol{\sigma}_i^z + \sum_{j \neq k=1}^N g_{j,k}\, \boldsymbol{\sigma}_j^z \otimes \boldsymbol{\sigma}_k^z,
\label{Ham}
\end{equation}
where $\boldsymbol{\sigma}_{\mathcal{S}}^z$ and $\boldsymbol{\sigma}_i^z$ are Pauli $z$ operators acting on the system and the $i$-th environment qubit, respectively. The coefficients $d_i$ represent the couplings between the system and the environment qubits, while $g_{j,k}$ denote the intra-environment couplings between the qubits $\mathcal{E}_j$ and $\mathcal{E}_k$. This model is illustrated in Fig.~\ref{alltoall}. 

\begin{figure}[h!]
\centering
\includegraphics[scale=0.38]{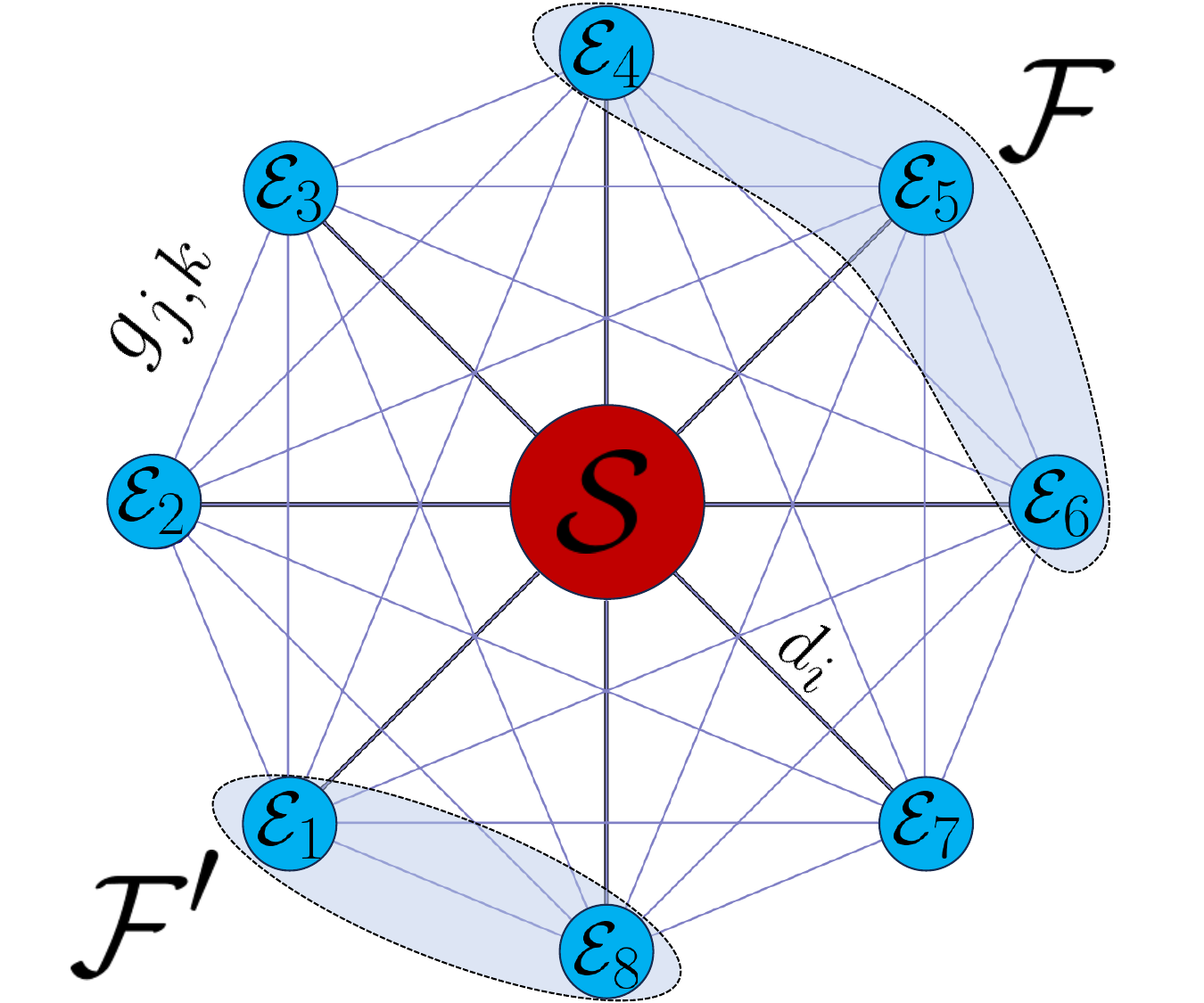}
\caption{\justifying Illustration of the all-to-all spin model with intra-environment interactions. Solid black lines connect $\mathcal{S}$ to each $\mathcal{E}_i$, representing the system-environment interactions governed by the Hamiltonian term $\boldsymbol{\sigma}_{\mathcal{S}}^z \otimes \sum_{i=1}^{N} d_i\,  \boldsymbol{\sigma}_i^z$. Thin, lighter blue lines connect every pair of environment qubits $\mathcal{E}_i$ and $\mathcal{E}_j$, illustrating the all-to-all intra-environment interactions described by the Hamiltonian term $\sum_{j \neq k=1}^{N} g_{j,k}\, \boldsymbol{\sigma}_j^z \otimes  \boldsymbol{\sigma}_k^z$.}
\label{alltoall}
\end{figure}

\begin{figure*}
    \vspace{-5.2mm}
    \includegraphics[width=\textwidth]{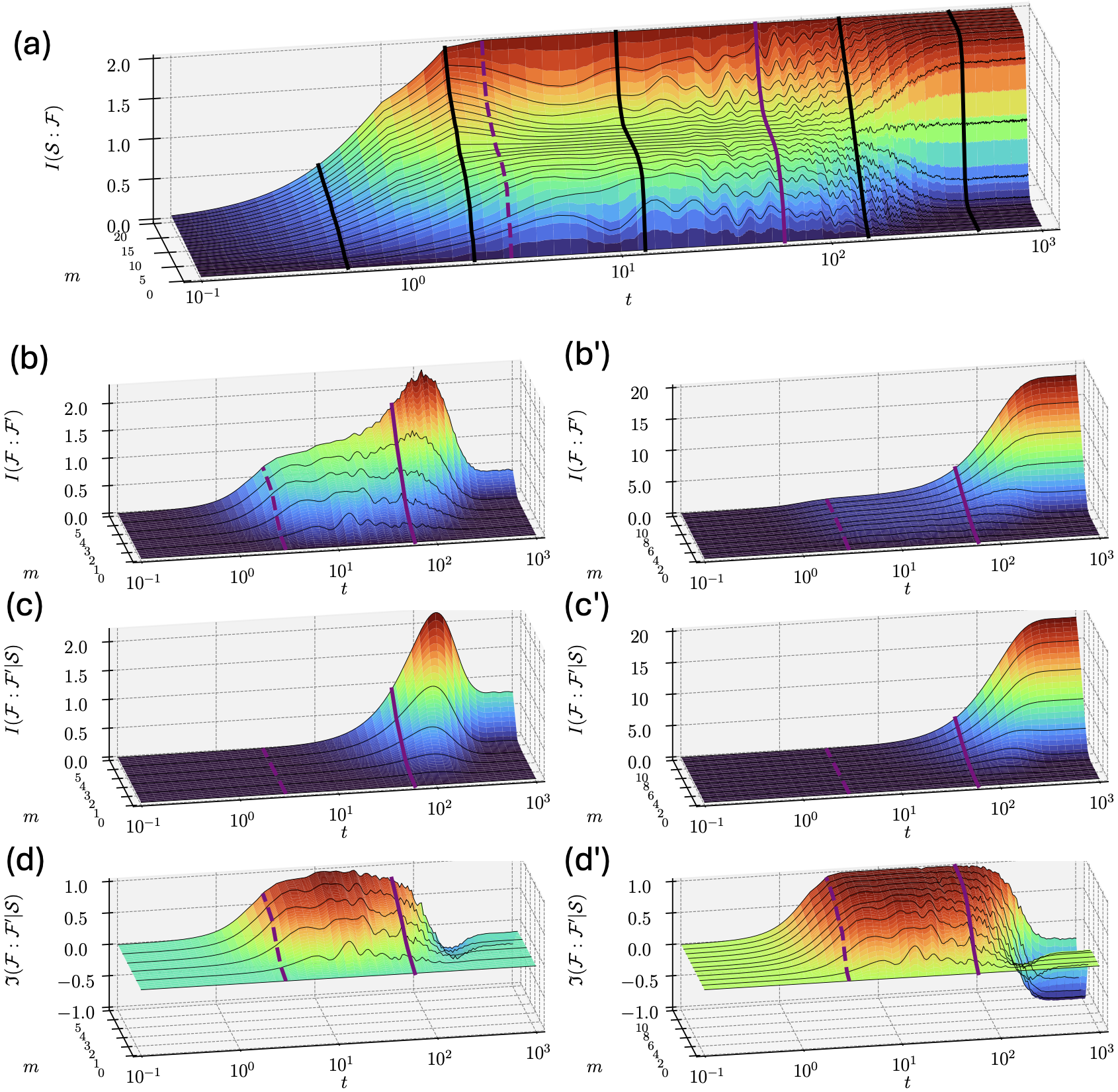}\vspace{0mm}
\caption{\justifying Evolution of mutual informations $I(\cS:\cF)$, $I(\cF:\cF’)$, the conditional $I(\cF:\cF’|\cS)$ and the refined ${\mathfrak I}(\cF:\cF’|\cS) = I(\cF:\cF) - I(\cF:\cF’|\cS)$. (a) Mutual information $I(\cS:\cF)$ increases quickly 
to the classical plateau level $H_\cS$ for sufficiently large fragments. That classical plateau level is reached after the decoherence time $t_{deco}$ when $\cE$, as a whole, acquires a record of the state of $\cS$. 
Quantum Darwinism regime starts when an almost complete information about $\cS$, $ I(\cS:\cF) \simeq (1-\delta) H_\cS $, is recorded in fractions $1/R_\delta$, $R_\delta \gg 2$ of $\cE$. The evolution of redundancy $R_\delta$ is seen 
in Fig.~\ref{combined1}. Quantum Darwinism regime ends as the interaction between the subsystems of $\cE$ leads to relaxation at time $t_{rel}$, scrambling the state of $\cS\cE$. This changes the nature of the dependence 
of the mutual information on the  fragment size $m$, as seen in the ``cuts'' in ($a$) as well as in Fig.~\ref{combined1} $b$. Mutual information $I(\cF:\cF’)$ and the conditional $I(\cF:\cF’|\cS)$ are seen in plots ($b$,$c$) for fragments $m \leq 5$, 
when decoherence by the ``rest of the environment’’ $\cE_{\cS\cF\cF’}$ can be still effective. However, as fragments become larger than that (see $b’,c’$), $\cE_{\cS\cF\cF’}$ is too small to impart decoherence. Note the sloping plateau in 
$I(\cF:\cF’)$ (best visible in ($b$), less visible in ($b’$) as a result of the scale change). Its extent coincides with the Quantum Darwinism regime. Within this time interval fragments of $\cE$ have compatible records of $\cS$, enabling consensus. 
Figures ($d,d'$) consensus $\mathfrak C$ defined using refined mutual information ${\mathfrak I}(\cF:\cF’|\cS) = I(\cF:\cF) - I(\cF:\cF’|\cS)$. Refined consensus has a pronounced plateau that coincides with the Quantum Darwinism 
regime where redundancy is significant. This classical plateau at the level $H_\cS$ persists even for fragment sizes $|\cF| > 5$, where decoherence is ineffective.}
\label{combined2}
\end{figure*}

\begin{figure*}
    \includegraphics[width=\textwidth]{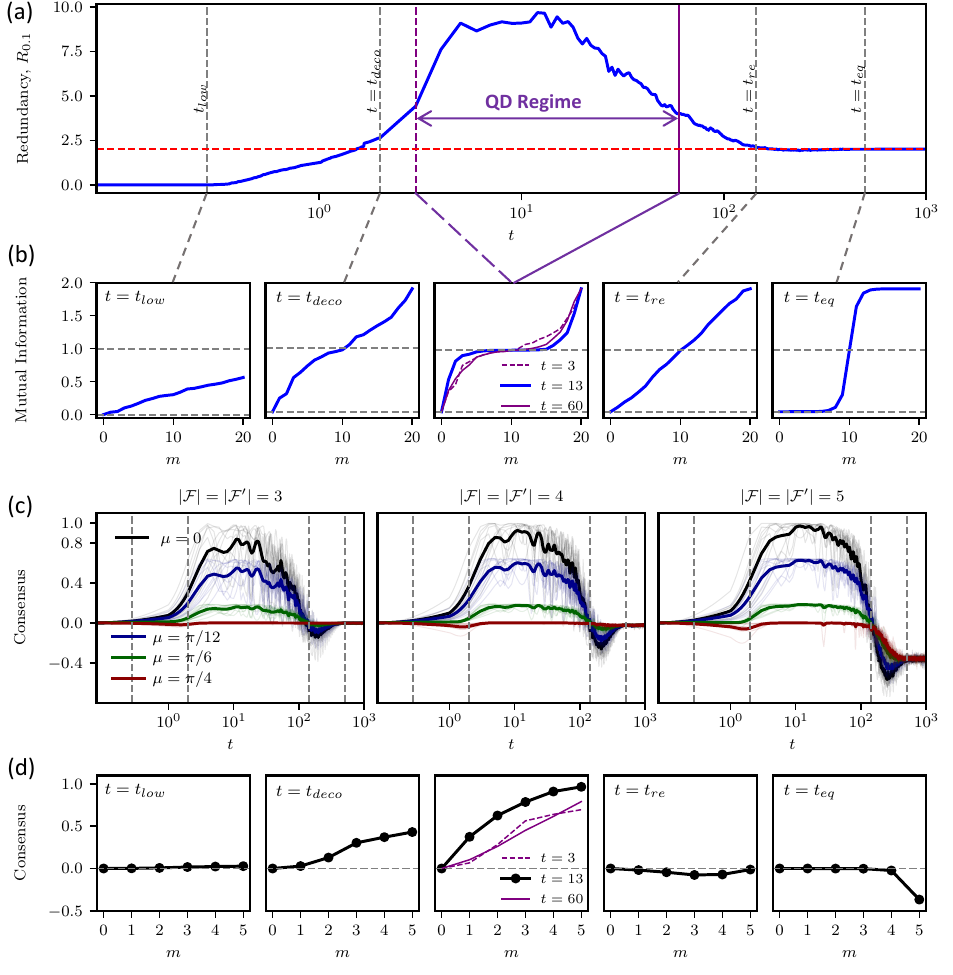}
\caption{\justifying  Redundancy, partial information plots (PIPs), and consensus for the model of Fig.~\ref{alltoall}. (a) Redundancy $R_\delta$ quantifies width of the plateau seen in Fig.~\ref{combined2}$a$. (b) ``Snapshot PIPs'' illustrate effects of decoherence, resulting buildup of redundancy, and eventually scrambling: Pronounced classical plateau persists in the Quantum Darwinism regime, but PIPs are initially linear in $m$: At $t\leq t_{deco}$ each subsystem of the environment contributes a bit more information about $\cS$. This changes in the Quantum Darwinism regime, when nearly complete information about $\cS$ can be recovered from fractions of $\cE$, as is seen in the persistence of the classical plateau. Eventually scrambling replaces records of $\cS$ with correlation to other fragments, resulting in the ``Haar random” looking PIPs that arise after $t_{eq}$. (c) Time dependence of refined mutual information for fragments of size $m=3,4,5$ and for observables with eigenstates tilted away from the pointer states by $\mu$. Thick lines represent averages over many runs, while thin lines are examples of individual realizations of the Hamiltonian, Eq.~\eqref{Ham}. (d) Consensus $\mathfrak C$ defined using refined mutual information ${\mathfrak I}(\cF:\cF’|\cS) = I(\cF:\cF) - I(\cF:\cF’|\cS)$ begins to rise near the decoherence time, persists for sufficiently large fragments for as long as there is a sizeable redundancy, and disappears as a consequence of scrambling and relaxation.}
\label{combined1}
\end{figure*}

The initial state of the combined system is: 
\begin{equation}
|\Psi_{\cS\cE}(0)\rangle = \frac{1}{\sqrt{2^{N+1}}} \left( |0\rangle + |1\rangle \right) \bigotimes^{N}_{i=1} \left( |0\rangle_i + |1\rangle_i \right),
\end{equation}
where $|0\rangle$ and $|1\rangle$ denote the eigenstates of $\boldsymbol{\sigma}_{\mathcal{S}}^z$, and $|{0}\rangle_i$ and $|{1}\rangle_i$ are the eigenstates of $\boldsymbol{\sigma}_i^z$ for each of $N=20$ environment qubits $\mathcal{E}_i$. A detailed discussion of the rise and fall of redundancy in this model can be found in Ref.~\cite{darwin5}. We use it to study the rise and fall of consensus.

The coupling constants $d_i$ and $g_{j,k}$ are chosen randomly from normal distributions with zero mean and standard deviations $\Delta_d$ and $\Delta_g$, respectively. A crucial assumption is that the system couples more strongly to each environment qubit than the environment qubits interact with one another, i.e., $\Delta_d \gg \Delta_g$. This assumption is valid in many physical scenarios, such as a photon bath where intra-environment interactions are negligible (effectively $\Delta_g = 0$). However, this may not hold in environments like a gas, where frequent collisions between molecules dominate. This disparity in interaction strengths 
ensures that decoherence induced by the environment happens faster than mixing due to intra-environment interactions. 

The results are illustrated, starting with Fig.~\ref{combined2}. Thus, Fig.~\ref{combined2}$a$ shows evolution of mutual information $I(\cS:\cF)$ as a function of time. After the initial pre-decoherence period (where the environment is only weakly correlated with $\cS$, so that the eigenstates of the density matrix of $\cS$ are not yet aligned with its pointer states), mutual information assumes the shape characteristic of Quantum Darwinism. Classical plateau appears soon after the decoherence time. Its length is approximately given by the redundancy $R_{\delta}$ (plotted separately in Fig.~\ref{combined1}$a$). The maximum of $R_{\delta}$ is reached at $t \approx 13$, but the pronounced classical plateau persists from at least $ t\leq 5$ until $t \geq 50$. Plateau disappears only when the mutual information due to the direct interaction between the subsystems of the environment begins to dominate $I(\cF:\cF’)$. This is seen in the thick lines marked on the diagram of the evolving mutual information, Fig.~\ref{combined2}$a$, and again in the ``snapshots’’ of $I(\cS:\cF)$ taken at the corresponding instants of time and shown in Fig~\ref{combined1}$b$.

The evolution of both $I(\cF:\cF’)$ and of the conditional mutual information $I(\cF:\cF’|S)$ is illustrated in the two sets of plots, Figs.~\ref{combined2}$b,c$ and \ref{combined2}$b’,c’$. The plots on the left, Figs.~\ref{combined2}$b-c$, show both of these quantities for fragment sizes less than 5. This cutoff on the fragment size is deliberate: It leaves 10 qubits for the reminder of the environment $\cE_{/{\cF\cF’}}$, so one can hope that the resulting decoherence, while imperfect, will suppress quantum correlations in $\cS\cF\cF’$. For fragments with sizes $|\cF|\sim 3 – 5$ the mutual information $I(\cF:\cF’)$ starts near 0, but rises to a sloping plateau after the decoherence time. That gently sloping plateau represents rise of the information shared by the environment subsystems. It persists until $t \geq 50$ (that is, for as long as there is a classical plateau of $I(\cS:\cF)$). After that, $I(\cF:\cF’)$ rises steeply, as the environment fragments become correlated through direct interactions after $t  \sim 100$, the time associated with onset of relaxation. 

The plot Fig.~\ref{combined2}$c$ of the conditional mutual information $I(\cF:\cF’|S)$ looks superficially similar, but there is one crucial difference: Unlike in $I(\cF:\cF’)$, there is no gently sloping plateau. There is however a peak, with the location that approximately matches the location of the peak of $I(\cF:\cF’)$. The peak in mutual information between the two fragments and the similar peak in the conditional mutual information indicate that the mutual information becomes dominated not by what they know about $\cS$, but by what $\cF$ and $\cF'$ know about one other, independently of the information they might share about $\cS$. 

The refined mutual information responsible for the consensus about the classical reality is given by their difference, ${\mathfrak I}(\cF:\cF’ |\cS)$,~Eq.~\eqref{refine}. Fig.~\ref{combined2}$d$, $d'$  shows the evolution of the consensus about the state of $\cS$ for various sizes $m=|\cF|$ of the environment fragments. That plot of ${\mathfrak I}(\cF:\cF’|\cS)$ distills what $\cF$ and $\cF’$ know about $\cS$ from what they know about one another (e.g., as a result of direct interaction). Its plateau extends approximately as far as the sloping plateau of $I(\cF:\cF’)$. For larger fragments ${\mathfrak I}(\cF:\cF’|\cS)$ reaches levels comparable to $H_\cS$, the missing classical information about $\cS$. This happens, however, only when the environment fragments are large enough. That is, the plateau is not there for single qubits, $|\cF|=1$, but is already quite visible for $|\cF| \geq 3 $.

Note that after the classical plateau disappears, the refined mutual information ${\mathfrak I}(\cF:\cF’|\cS)$ decreases, and for large fragments it can even dip to negative values. This occurs when the conditional mutual information $I(\cF:\cF’|S)$ is larger than $I(\cF:\cF’)$. In other words, it occurs when conditioning (here, on $\cS$) increases the mutual information beyond what it was before conditioning. This is a familiar behavior that may well have classical origins (cf.~\cite{CoverThomas}), but in our case it may be in part quantum: When the fragments are sufficiently large, the reminder of the environment $\cE_{/{\cF\cF’}}$ becomes too small to induce decoherence in $\cS\cF\cF’$. This starts already when $|\cF|=|\cF’| \geq 4$, especially at late times, after all the qubits of $\cS\cE$ become entangled. This is also seen in the evolution of the information responsible for the consensus, ${\mathfrak I}(\cF:\cF’|\cS)$, in Figs.~\ref{combined2}$c$, and~\ref{combined2}$c'$.

The rightmost column, Figs.~\ref{combined2}$b^{\prime}$-~\ref{combined2}$c'$, confirms the above narrative regarding the origin of objective classical reality. Figs.~\ref{combined2}$b’$-~\ref{combined2}$c'$, also show how and to what extent it breaks down when the reminder of the environment $\cE_{/{\cF\cF’}}$ becomes gradually too small to effectively decohere $\cS\cF\cF’$ (so that the correlations are still somewhat quantum). This is even more pronounced as $|\cF|=|\cF’|$ approaches 10, at which point there is no decohering environment left. The most striking difference with the plots, Figs.~\ref{combined2}b-d, is the rise of both $I(\cF:\cF’)$ and $I(\cF:\cF’|S)$ to a much higher peak value, its size approximately corresponding to the dimensionality (hence, entropy) of $|\cF|=|\cF’|$. The two peaks in Figs.~\ref{combined2}$b^{\prime}$ and~\ref{combined2}$c'$ have nearly the same size, so what the fragments know about one another is much more than the missing information (i.e., the entropy of $H_\cS$). The ``sloping plateau’’ of Fig.~\ref{combined2}$b$ is still there in Fig.~\ref{combined2}$b^{\prime}$, but the change of the scale imposed by the size of the peak makes it harder to see.

In spite of the domination of the information acquired by the direct interaction between the two fragments, the difference that defines the refined mutual information, ${\mathfrak I}(\cF:\cF’|\cS)$, in Fig.~\ref{combined2}$d’$, looks very similar to Fig.~\ref{combined2}$d$, in that the classical plateau extends over a similar range---that is, starting after the decoherence time, and ending around the relaxation time. Thus, consensus defined with the help of the refined mutual information, ${\mathfrak I}(\cF:\cF’|\cS)$, concerns the state of the system alone, as it should. 

Several ``snapshots’’ of consensus are plotted as a function of the fragment size in Fig.~\ref{combined1}$c$. As was already pointed out before, consensus is present for sufficiently large fragments in the same time interval where classical plateau is also present.

Fig.~\ref{combined1}$c$ explores the time dependence of consensus on the observable of the system. Refined mutual information is plotted there, as a function of time, for the pointer observable $\sigma_z$ as well as for the several other observables with the eigenstates tilted away from the pointer states. Clearly, and as expected, consensus is maximized for the pointer observable, and is gradually diminished as the tilt increases.

Finally, Fig.~\ref{combined1}d shows that the refined consensus is present when the information about $\cS$ is imprinted redundantly in the environment. However, consensus is absent before and is still rather modest at decoherence time. Moreover, scrambling (which replaces the information fragments of the environment had about $\cS$ with the information about the other fragments) destroys consensus and established correlations that can lead to negative refined information ${\mathfrak I}(\cF:\cF’|\cS)$.

An important additional lesson follows from the above discussion, and especially from the plots that illustrate the rise and fall of redundancy seen in Figs.~\ref{combined2}$a$,~\ref{combined1}$a$,~\ref{combined1}$c$, and all the other plots in this section: Redundancy is necessary for consensus---by definition, several fragment must have consistent information about the system to agree (to arrive at consensus) about its state. Therefore, consensus is possible only after the decoherence process disseminates multiple copies of the information about the system in the environment. Moreover, redundancy is destroyed by scrambling of information that replaces what the fragments of $\cE$ know about $\cS$ with what they know about one another. Therefore, consensus can persist only until the relaxation time when the scrambling of information induced by the interactions between the subsystems of $\cE$ replaces what they know about $\cS$ with what they know about each 
other.

\section{Concluding remarks} 

The mystery of the emergence of the classical world of our everyday experience from within the Universe that is fundamentally quantum (as it is made out of quantum components) has been with us for 100 years. We show that it is resolved by recognizing the role of the environment in the emergence of ``the classical’’. The process starts with environment-induced decoherence which einselects preferred stable states, pointer states immune to the interaction with the environment. Thus, superpositions of pointer states are eliminated, and the eigenvalues of the reduced density matrix can be regarded as their probabilities. 

Even after the system has decohered there are still many such stable pointer states on the diagonal of the resulting mixture. So, the question about the perception of a single outcome is not addressed by decoherence alone: Einselection supplies the ``menu’’ of the candidate quasiclassical states, but it does not explain why only one of them is always seen. 

Quantum Darwinism---the recognition that also perception by observers relies on the environment, which in this case acts as the communication channel---explains why our perceptions are always consistent with a unique outcome and, hence, why we attribute objective existence to the states of the everyday objects of interest in everyday settings. 

Environment acts as a witness to the state of the system, each of its fragments holding a record of the systems state. Redundancy with which a state of the decohering system is imprinted on the environment plays crucial role. On the one hand, imprinting of multiple (even imperfect) copies in the environment constrains the set of states that can act as ``originals’’: They must be distinguishable (ideally, orthogonal) to survive creation of redundant records in the environment, and, hence to be perceived. The result is that only the environment - resistant pointer states can be found out indirectly, via their imprints in the fragments of $\cE$. This is suggestive of the ``wavepacket collapse’’, as it precludes the possibility perceiving the pre-measurement state and limits what can be observed to the pointer states.

Redundancy allows for re-confirmation, and it lets many access the information about the einselected state. Our focus was on consensus, as it implies collapse of the evidence about the pointer state of the system perceived by intercepting fragments of its environment. Indeed, the only evidence of the wavepacket collapse is this consensus. 

We have used mutual information between the fragments of the environment as the key ingredient of the measure of consensus. The first key result is Theorem 1 which shows that, when fragments of the environment that interacted with the system have sufficient information about its state, they will inevitably agree on what that state is. Hence, redundancy implies the collapse of the evidence: The data in every environment fragment with enough information about the system point to the same einselected pointer state. 

This collapse of evidence is a firm, Quantum Darwinism - based prediction of quantum theory. It does not require, prior to perception, a unique state of the object. Rather, it implies that redundant system-environment subsystems correlations will necessarily point to the unique state of the system. This collapse of evidence is tantamount to the wavepacket collapse, even though it is not preceded by a literal collapse.

These conclusions follow from Theorem 1. They are supported by the analytically solvable {\tt c-maybe} model. The key assumption of Theorem 1 is the branching structure [exemplified by Eq.~\eqref{branches}]. Branching arises, via decoherence, when the subsystems of the environment monitor the pointer observable of the system of interest, but do not interact with each other. The photon environment is far from equilibrium (e.g., sunlight), satisfies this assumption, and is the usual communication channel through which we perceive our everyday world.

We have also proposed a measure of consensus between the environment fragments suitable for situations when the subsystems of $\cE$ become correlated (so the information they agree about may arise through direct interactions between them, and may have, therefore, little to do with the system of interest). In this case we have suggested the refined mutual information as a suitable measure: It helps filter the information about the system from spurious, scrambling-induced correlations between the environment fragments. 

Refined mutual information, ${\mathfrak I} (\cF:\cF’|\cS)$, (also known as ``interaction’’ to information theorists) filters out such spurious information, as was demonstrated numerically in a many-body model. There, the interactions between subsystems of $\cE$ eventually scrambled the state of $\cS\cE$, changing the nature of the partial information plots. Refined mutual information separated the relevant and irrelevant contributions to consensus.

However, refined mutual information has an unusual feature: It can be negative. In the setting involving system $\cS$ and two environment fragments this happens when the information gain about $\cS$ increases mutual information between $\cF$ and $\cF’$. The cause can be either increase of classical correlation (see~\cite{CoverThomas}) or (in our quantum case) entanglement (when the measurement of $\cS$ reveals that $\cF$ and $\cF’$ are entangled). We have seen this happen as a result of scrambling in the numerical model.

Theorem 2 shows that, for branching states, refined mutual information is always nonnegative providing that the branches are decohered (e.g., by the rest of the environment). Thus, for the far-from-equilibrium fragments of the photon environment (which human observers use to gain vast majority of what they know about the world), refined mutual information is a useful measure of consensus. Moreover, for branching states ${\mathfrak I} (\cF:\cF’|\cS)$ coincides with mutual information between the fragments. 

The situation becomes more complicated when there are preexisting correlations between the environment subsystems, or when subsystems interact, so the state is no longer branching. Detailed analysis of such cases is beyond the scope of this work. As we have seen in Figs.~\ref{combined2},~\ref{combined1}, when the environment subsystems interact and / or decoherence is ineffective, refined mutual information can be negative. Nevertheless, as Figs.~\ref{combined2},~\ref{combined1} also illustrate, as long as decoherence is effective (e.g., the reminder of the environment suffices to decohere branches of $\cS\cF\cF’$) ${\mathfrak I} (\cF:\cF’|\cS)$ or ${\mathfrak C} (\cF:\cF’|\cS)$ is an effective measure of consensus, as it can distinguish consensus about $\cS$ even when the conditions that assure positivity of refined mutual information are not met.

We end by noting that the analysis of the emergence of consensus sheds new light, and perhaps even settles, the issue of the wavepacket collapse. The answer that emerges may be viewed by some as unexpected: Rather than the ``literal collapse’’ into a single state of the system (the goal of, e.g., the spontaneous collapse program), consensus guarantees---using information-theoretic tools---the ``collapse’’ of perceptions: Consensus addresses the question of the perception of a single outcome---the origin of the collapse of the evidence. This is all that is needed to explain the origin of the objective classical reality.

\acknowledgements{We thank Sebastian Deffner and Davide Girolami for discussions and comments. A.T. acknowledges support from the Center for Nonlinear Studies and the U.S. Department of Energy under the LDRD program at Los Alamos. B.Y. acknowledges support in part from the U.S. Department of Energy, Office of Science, Office of Advanced Scientific Computing Research, through the Quantum Internet to Accelerate Scientific Discovery Program, and in part from the LDRD program at Los Alamos.}

\bibliography{opm}

\onecolumngrid

\newpage

\begin{center}
\textbf{\large Supplementary material \vspace{0.05in} \\
Consensus About Classical Reality in a Quantum Universe}
\end{center}
In the following supplementary material we first show the derivation of the analytic expressions for the information deficits in the {\tt c-maybe}. We then present the full proofs of Theorem 2 stated in the main text.

\setcounter{equation}{0}
\setcounter{figure}{0}
\setcounter{table}{0}
\setcounter{page}{1}
\makeatletter
\renewcommand{\theequation}{S\arabic{equation}}
\renewcommand{\thefigure}{S\arabic{figure}}
\renewcommand{\bibnumfmt}[1]{[S#1]}
\renewcommand{\citenumfont}[1]{#1}

\section*{Information deficits and finite-size environment in the {\tt c-maybe} model}

We consider the {\tt c-maybe} model and we derive analytic expressions for the mutual information $I(\cS:\cF)$ and $I(\cF:\cF^{'})$ near the plateau in the good decoherence limit ($N$ tends to infinity). It will prove useful to work with the series expansions of the quantities of interest in our problem. In fact we know that, for $|x|<1$,
\begin{equation}
	\log_{2}(1+x)=\left(\frac{1}{\ln(2)}\right)\sum_{i=1}^{\infty}(-1)^{i+1} \frac{x^{i}}{i}, \ \ \text{and} \ \ \log_{2}(1-x)=-\left(\frac{1}{\ln(2)}\right) \sum_{i=1}^{\infty}  \frac{x^{i}}{i}.
\end{equation}
In the {\tt c-maybe} model, for a system starting in the initial state $|\psi_{S}\rangle_{0}= \sqrt{p} |0\rangle + \sqrt{q} |1\rangle$, we can derive the analytic expression for $I\left(\mathcal{S}: \mathcal{F}_m\right)$ (Touil et al.):
$$
I\left(\mathcal{S}: \mathcal{F}_m\right)=h\left(\lambda_{N, p}^{+}\right)+h\left(\lambda_{m, p}^{+}\right)-h\left(\lambda_{N-m, p}^{+}\right),
$$
where $h(x)=-x \log_{2}(x)-(1-x) \log_{2}(1-x)$ and
$$
\lambda_{k, p}^{ \pm}=\frac{1}{2}\left(1 \pm \sqrt{(q-p)^2+4 s^{2 k} p q}\right) .
$$

We thus have a closed expression for the mutual information $I\left(\mathcal{S}: \mathcal{F}_m\right)$.

From the above, for $p=q=1/2$, and for infinite environment ($N$ very large), we get
$$
I\left(\mathcal{S}: \mathcal{F}_m\right)\approx h\left(\lambda_{m, p}^{+}\right),
$$
and
\begin{equation}
h\left(\lambda_{m, p}^{+}\right)\approx h\left(\frac{1}{2}(1+s^{m})\right),
\end{equation}
hence using the expansion in equation (1) we get
\begin{equation}
	\log_{2}(1+x) \approx \left(\frac{1}{\ln(2)}\right)(x-\frac{x^2}{2}), \ \ \text{and} \ \ \log_{2}(1-x)\approx \left(\frac{1}{\ln(2)}\right)(-x-\frac{x^2}{2}),
\end{equation}
which imply
\begin{equation}
\begin{split}
I\left(\mathcal{S}: \mathcal{F}_m\right) &\approx \log_{2}(2)-\left(\frac{1}{\ln(2)}\right) \frac{s^{2m}}{2}\\
& = H_{\mc{S}}- \left(\frac{e^{H_{\mc{S}}}-1}{2H_{\mc{S}}}\right) s^{2m}.
\end{split}
\end{equation}

and 
\begin{equation}
I\left(\mathcal{F}_m: \mathcal{F}_{m^{'}}\right) \approx H_{\mc{S}} - \left(\frac{e^{H_{\mc{S}}}-1}{2H_{\mc{S}}}\right) \left(s^{2m}+s^{2m^{'}}-s^{2(m+m^{'})}\right).
\label{main01}
\end{equation}

For the case where $p\neq q$ (without loss of generality we assume $q>p$, to avoid carrying the absolute value everywhere), we get
$$
\lambda_{k, p}^{ \pm} \approx q\left(1 + s^{2k} \frac{p}{q-p}\right) .
$$
Which implies
\begin{equation}
I\left(\mathcal{S}: \mathcal{F}_m\right) \approx -q\log_{2}(q) - p\log_{2}(p) - \frac{pqs^{2m}}{q-p}\log_{2}(\frac{q}{p})- q(1+\frac{ps^{2m}}{q-p})\log_{2}(1+\frac{ps^{2m}}{q-p})-p(1-\frac{qs^{2m}}{q-p})\log_{2}(1-\frac{qs^{2m}}{q-p}),
\end{equation}
using the expansion in equation (1) we get 
\begin{equation}
I\left(\mathcal{S}: \mathcal{F}_m\right) \approx  H_{\mc{S}}- \frac{pq}{q-p}\log_{2}(\frac{q}{p})s^{2m}.
\end{equation}
and 
\begin{equation}
I\left(\mathcal{F}_m: \mathcal{F}_{m^{'}}\right) \approx  H_{\mc{S}}- \frac{pq}{q-p}\log_{2}(\frac{q}{p})\left(s^{2m}+ s^{2m^{'}}-s^{2(m+m^{'})} \right).
\end{equation}
In the limit where $p=q$ we recover the expression in Eq.~\eqref{main01}.

Now, we assume that $N$ is finite and for the case where $p\neq q$ (similar to the above analysis), we have
$$
\lambda_{k, p}^{ \pm} \approx q\left(1 + s^{2k} \frac{p}{q-p}\right) .
$$
hence
\begin{equation}
H_{X} =-q\log_{2}(q) - p\log_{2}(p) - \frac{pqs^{2k}}{q-p}\log_{2}(\frac{q}{p})- q(1+\frac{ps^{2k}}{q-p})\log_{2}(1+\frac{ps^{2k}}{q-p})-p(1-\frac{qs^{2k}}{q-p})\log_{2}(1-\frac{qs^{2k}}{q-p}),
\end{equation}
where $H_{X}$ is the entropy of subsystem $X$. For $H_{\mc{S}}$ ($H_{\mc{SF}}$), we have $k=N$ ($k=N-m$). And for $H_{\mc{SF}'}$ ($H_{\mc{SFF'}}$), we have $k=N-m'$ ($k=N-m-m'$). Using the expansion in equation (1) we get
\begin{equation}
H_{X} =H^{max}_{\mc{S}}- \frac{pq}{q-p}\log_{2}(\frac{q}{p})s^{2k},
\end{equation}
such that $H^{max}_{\mc{S}}=-q\log_{2}(q) - p\log_{2}(p)$ is the maximum entropy of the system (only achieved in the limit of good decoherence). Therefore, we get
\begin{equation}
\begin{split}
H_{\mc{S}} &=H^{max}_{\mc{S}}- \frac{pq}{q-p}\log_{2}(\frac{q}{p})s^{2N},\\
H_{\mc{SF}} &=H^{max}_{\mc{S}}- \frac{pq}{q-p}\log_{2}(\frac{q}{p})s^{2(N-m)},\\
H_{\mc{SF}'} &=H^{max}_{\mc{S}}- \frac{pq}{q-p}\log_{2}(\frac{q}{p})s^{2(N-m')},\\
H_{\mc{SFF}'} &=H^{max}_{\mc{S}}- \frac{pq}{q-p}\log_{2}(\frac{q}{p})s^{2(N-m-m')}.
\end{split}
\end{equation}
Hence we can identity
\begin{equation}
\begin{split}
\epsilon &=\frac{pq}{H_{\cS}(q-p)}\log_{2}(\frac{q}{p})\left(s^{2N}-s^{2(N-m)}\right),\\
\epsilon' &=\frac{pq}{H_{\cS}(q-p)}\log_{2}(\frac{q}{p})\left(s^{2N}-s^{2(N-m')}\right),\\
\tilde{\epsilon} &=\frac{pq}{H_{\cS}(q-p)}\log_{2}(\frac{q}{p})\left(s^{2N}-s^{2(N-m-m')}\right).\\
\end{split}
\end{equation}
In the limit where $p=q=1/2$, we get
\begin{equation}
\begin{split}
\epsilon &=\left(\frac{e^{H^{max}_{\mc{S}}}-1}{2H_{\cS}}\right)\left(s^{2N}-s^{2(N-m)}\right),\\
\epsilon' &=\left(\frac{e^{H^{max}_{\mc{S}}}-1}{2H_{\cS}}\right)\left(s^{2N}-s^{2(N-m')}\right),\\
\tilde{\epsilon} &=\left(\frac{e^{H^{max}_{\mc{S}}}-1}{2H_{\cS}}\right)\left(s^{2N}-s^{2(N-m-m')}\right).\\
\end{split}
\end{equation}

\section*{Consensus for singly-branching states}

\subsection*{Theorem statement}
\paragraph*{\textbf{Theorem 2:}} \textit{Consider the singly-branching structure of the global wave function $|\psi_{\mc{SE}}\rangle$ in $D_{\cS}$-dimensions, such that
\begin{equation}
|\psi_{\mc{SE}}\rangle= \sum^{D_{\mc{S}}-1}_{n=0} \sqrt{q_n} |s_n\rangle |\mc{F}_n\rangle |\mc{F}'_n\rangle |\cE_{/\cF \cF'_n}\rangle,
\end{equation}
where $\{|s_n\rangle\}_{n \in \llbracket 1; D_{\cS} \rrbracket}$, $\{|\mc{F}_n\rangle\}_{n \in \llbracket 1; D_{\cS} \rrbracket}$, and $\{|\cE_{/\cF \cF'_n}\rangle\}_{n \in \llbracket 1; D_{\cS} \rrbracket}$ are orthonormal states. Therefore, independent of the measurements applied on $\mc{S}$ we have
\begin{equation}
{\mathfrak I }(\cF : \cF' | \cS)=I(\mc{F}:\mc{F}')-I(\mc{F}:\mc{F}'|{\mc{S}}) \geq 0.
\end{equation}
}

\paragraph*{\textbf{Proof of Theorem 2:}} The goal here is to prove that any measurement on singly-branching states leads to positive consensus, for orthogonal fragment states $\mc{F}$ and orthogonal states of $R$ (i.e. the rest of the environment is also large enough to ensure orthogonality). The structure of states we consider is the following
\begin{equation}
|\psi_{\mc{SE}}\rangle= \sum^{D_{\mc{S}}}_{n=1} \sqrt{q_n} |s_n\rangle |\mc{F}_n\rangle |\mc{F}'_n\rangle|\cE_{/\cF \cF'_n}\rangle,
\end{equation}
In the above notation, subsystem $\cE_{/\cF \cF'}$ is the rest of the environment. From this expression we get
\begin{equation}
\rho_{\mc{F}}= \sum_{n} q_n |\mc{F}_n\rangle \langle \mc{F}_n|,
\end{equation}
\begin{equation}
\rho_{\mc{F}'}= \sum_{n} q_n |\mc{F}'_n\rangle \langle \mc{F}'_n|,
\end{equation}
\begin{equation}
\rho_{\mc{FF}'}= \sum_{n} q_n |\mc{F}_n\rangle \langle \mc{F}_n| \otimes |\mc{F}'_n\rangle \langle \mc{F}'_n|.
\end{equation}
Since $|\mc{F}_n\rangle$ form an orthonormal basis, we have~\cite{nielsen2002quantum}
\begin{equation}
H_{\mc{FF}'}= H[q_n]+\sum_{n} q_n H_{|\mc{F}'_{n}\rangle}= H[q_n]+ 0= H_{\mc{F}},
\end{equation}
where $H[q_n]=-\sum_{n} q_n \log_{2}{(q_n)}$. This directly implies that $I(\mc{F}:\mc{F}')= H_{\mc{F}'}$.

Now, let's determine the $I(\mc{F}:\mc{F}'|{\mc{S}})= \sum_i \gamma_i {\mathfrak D}_{\text{KL}}(\rho^{(i)}_{\mc{F}\mc{F}'}||\rho^{(i)}_{\mc{F}} \otimes \rho^{(i)}_{\mc{F}'})$ for arbitrary POVMs on $\mc{S}$~\footnote{We consider complete measurements, hence $\sum_{i} \gamma_i=1$}. Using Stinespring's dilation we can enlarge our Hilbert space and model arbitrary POVMs with a unitary between $|\psi_{\mc{SE}}\rangle$ and an observer $|O_0\rangle$. In particular, the conditional states of the fragments ($\mc{F}$, $\mc{F}'$, and $\mc{F}\mc{F}'$) are defined with respect to some reference observer states $\{|\mc{O}_i\rangle\}_i$. A measurement in this case is a unitary $U$ coupling the system $|\mc{S}\rangle$ and an observer initially in state $|O_0\rangle$. We then get
\begin{equation}
U|\psi_{\mc{SE}}\rangle |O_0\rangle=\sum_n \sqrt{q_n} \left( U |s_n\rangle |O_0\rangle\right)  |\mc{F}_n\rangle |\mc{F}'_n\rangle |\cE_{/\cF \cF'_n}\rangle,
\end{equation}
which implies
\begin{equation}
\begin{split}
&U|\psi_{\mc{SE}}\rangle |O_0\rangle =\sum_{n,j} \sqrt{q_n} \sqrt{c^{(n)}_j} |n_j\rangle_{\cS}  |\mc{F}_n\rangle |\mc{F}'_n\rangle  |O^{n}_j\rangle |\cE_{/\cF \cF'_n}\rangle,\\
&= \sum_i \sqrt{\gamma_i}\left(\sum_{n,j} \frac{\sqrt{q_n} \sqrt{c^{(n)}_j}}{\sqrt{\gamma_i}} \langle \mc{O}_i  |O^{n}_j\rangle   |n_j\rangle_{\cS} |\mc{F}_n\rangle |\mc{F}'_n\rangle  |\cE_{/\cF \cF'_n}\rangle\right)\otimes |\mc{O}_i\rangle,
\end{split}
\label{s11}
\end{equation}
For convenience of notation let's define $\sqrt{\beta_{i,j,n}} \equiv \sqrt{c^{(n)}_j} |\langle \mc{O}_i  |O^{n}_j\rangle|$. Then, the conditional reduced density matrices read
\begin{equation}
\rho^{(i)}_{\mc{F}}= \sum_{n} \frac{q_n\sum_{j}\beta_{i,j,n}}{\gamma_i} |\mc{F}_n\rangle \langle \mc{F}_n|,
\end{equation}
\begin{equation}
\rho^{(i)}_{\mc{F}'}= \sum_{n} \frac{q_n\sum_{j}\beta_{i,j,n}}{\gamma_i} |\mc{F}'_n\rangle \langle \mc{F}'_n|,
\end{equation}
\begin{equation}
\rho^{(i)}_{\mc{FF}'}= \sum_{n} \frac{q_n\sum_{j}\beta_{i,j,n}}{\gamma_i} |\mc{F}_n\rangle \langle \mc{F}_n| \otimes |\mc{F}'_n\rangle \langle \mc{F}'_n|,
\end{equation}
where $\gamma_i=\sum_{n,j} q_n \beta_{i,j,n}$. Similar to the above argument we have
\begin{equation}
H^{(i)}_{\mc{FF}'}= H[\frac{q_n\sum_{j}\beta_{i,j,n}}{\gamma_i}]+\sum_{n} \frac{q_n\sum_{j}\beta_{i,j,n}}{\gamma_i} H^{(i)}_{|\mc{F}'_n\rangle}= H[\frac{q_n\sum_{j}\beta_{i,j,n}}{\gamma_i}]+ 0= H^{(i)}_{\mc{F}},
\label{entorpyexp}
\end{equation}
hence $I(\mc{F}:\mc{F}'|{\mc{S}})= \sum_i \gamma_i H^{(i)}_{\mc{F}'}$. Therefore, since the von Neumann entropy is concave, we get
\begin{equation}
{\mathfrak I }(\cF : \cF' | \cS)= H_{\mc{F}'}-\sum_i \gamma_i H^{(i)}_{\mc{F}'} \geq 0.
\end{equation}
QED.

\end{document}